\documentclass[11pt,a4paper]{article}

\usepackage[T1]{fontenc}
\usepackage[utf8]{inputenc}
\usepackage[a4paper,margin=25mm]{geometry}
\usepackage{amsmath,amssymb,amsthm}
\usepackage{graphicx}
\usepackage{natbib}
\usepackage{longtable,booktabs,makecell}
\usepackage{xurl}
\usepackage{microtype}

\newtheorem{theorem}{Theorem}
\newtheorem{proposition}[theorem]{Proposition}
\newtheorem{biomlemma}[theorem]{Lemma}

\title{Validity of MMRM-based hypothesis testing under missing-not-at-random mechanisms}

\author{
Kazushi Maruo$^{1,\dagger,*}$,
Ryota Ishii$^{1,\dagger,*}$,
Yusuke Yamaguchi$^{2}$,\\
Keisuke Hanada$^{3}$,
Naoki Isogawa$^{4}$,
and Masahiko Gosho$^{1}$\\[0.6em]
\small $^{1}$Department of Biostatistics, Institute of Medicine, University of Tsukuba,\\
\small 1-1-1 Tennodai, Tsukuba, Ibaraki 305-8575, Japan\\
\small $^{2}$Quantitative Sciences and Evidence Generation, Astellas Pharma Global Development Inc.,\\
\small 2375 Waterview Drive, Northbrook, Illinois 60062, USA\\
\small $^{3}$Department of Biostatistics, Wakayama Medical University,\\
\small 881-1 Kimiidera, Wakayama 641-8509, Japan\\
\small $^{4}$Biometrics \& Data Science, UCB Japan Co., Ltd.,\\
\small Shinjuku Grand Tower 13F, 8-17-1 Nishi-Shinjuku, Shinjuku-ku, Tokyo 160-0023, Japan\\[0.6em]
\small $^{\dagger}$These authors contributed equally to this work.\\
\small $^{*}$Co-corresponding authors:
Kazushi Maruo (\texttt{maruo@md.tsukuba.ac.jp});\\
Ryota Ishii (\texttt{rishii@md.tsukuba.ac.jp})
}

\date{}

\begin{document}

\maketitle

\begin{abstract}
In randomized clinical trials with longitudinal continuous outcomes, missing-not-at-random (MNAR) missingness often motivates conservative alternatives to mixed models for repeated measures (MMRM). Such caution is important for estimation, but estimation and testing need not require identical assumptions. Moreover, overly conservative primary analyses may reduce power, increase required sample size, and raise trial costs. We investigated the validity of MMRM-based testing under the global null of identical longitudinal outcome distributions across groups. Because valid testing minimally requires treatment-effect estimators to converge to the null under the null hypothesis, we investigated sufficient conditions for this property. We introduced a proportional observation condition requiring ratios of observation probabilities relative to a reference group, conditional on the full outcome vector, to be outcome-independent, and showed that, with arbitrary post-baseline visits and monotone missingness, this condition is sufficient for convergence to the null value. The condition allows observation to depend on unobserved outcomes and permits between-group differences in overall observation probabilities through outcome-independent dropout, making it clinically interpretable while accommodating outcome-dependent MNAR missingness. Synthetic and data-based bootstrap simulations showed negligible bias and empirical test sizes near 0.05, including nonmonotone missingness. Thus, MNAR missingness does not by itself imply that a more conservative primary testing procedure is required. This result does not justify treatment-effect estimation under alternatives, which still requires estimand-based interpretation and sensitivity analyses.
\end{abstract}


\section{Introduction}

Missing data are almost unavoidable in clinical trials in which outcomes
are measured repeatedly over time. Participants may discontinue follow-up
because of insufficient efficacy, worsening disease status, adverse events,
treatment burden, or other reasons, and the resulting missing outcomes can
affect the validity and interpretation of statistical analyses. Since the
publication of the ICH E9(R1) addendum, the handling of missing data has
increasingly been considered in relation to estimands, intercurrent events,
and sensitivity analyses \citep{ich_e9r1}. In particular, increasing
attention has been paid to whether the assumptions underlying the primary
analysis are compatible with the estimand of interest and whether the
robustness of conclusions to departures from those assumptions has been
adequately assessed.

The mixed model for repeated measures (MMRM) is widely used for the primary analysis
of longitudinal continuous outcomes in randomized clinical trials
\citep{mallinckrodt_etal_01}. The MMRM method uses available post-baseline
measurements without explicit imputation of missing outcomes. 
A commonly used implementation of MMRM for longitudinal continuous outcomes in randomized clinical trials combines restricted maximum likelihood estimation with an unstructured within-subject covariance matrix and Kenward--Roger inference for fixed effects \citep{kenward_roger_97}.
Under the missing-at-random (MAR) assumption, likelihood-based inference using the
MMRM method provides valid inference for mean treatment differences at
prespecified visits \citep{little_rubin}. When missingness depends on
unobserved outcomes, however, the missing-data mechanism may be missing not
at random (MNAR), and this likelihood-based justification no longer
generally applies \citep{little_rubin, molenberghs_kenward, gosho_etal_17}.

A variety of approaches have therefore been proposed to assess the
consequences of departures from MAR, including pattern-mixture models,
selection models, delta-adjustment methods, reference-based imputation, and
tipping-point analyses
\citep{little_93,little_95,molenberghs_kenward,
carpenter_kenward,carpenter_etal_13,mallinckrodt_etal_13}.
These approaches are particularly important when the objective is to
estimate and interpret the magnitude of a treatment effect under a specified
estimand. More generally, concern about possible MNAR missingness tends to
favor conservative assumptions and analysis strategies, because the
distribution of unobserved outcomes cannot be identified from the observed
data alone. Such caution is often appropriate for treatment effect
estimation. However, excessive conservatism can reduce statistical power,
increase the required sample size, and adversely affect the feasibility and
efficiency of confirmatory clinical trials.

The requirements for valid treatment effect estimation and those for valid
hypothesis testing are not necessarily identical. In a confirmatory trial,
a fundamental requirement of a hypothesis test is control of the type I
error rate. This concerns the behavior of the test under the null hypothesis,
rather than the bias or interpretability of the treatment effect estimator
under every possible alternative. Thus, even when the MMRM method cannot
generally be justified for treatment effect estimation under MNAR missingness,
MMRM-based hypothesis testing may remain valid under some MNAR mechanisms.
If such mechanisms can be characterized, the mere possibility of MNAR
missingness need not necessarily imply that a more conservative testing
procedure is required.

In this study, we focus on the global null hypothesis that the full joint
distribution of the longitudinal outcomes is identical between treatment
groups. This hypothesis is stronger than a pointwise null hypothesis
specifying a zero mean difference at a particular visit. We investigate
whether there exist MNAR mechanisms under which, despite dependence of the
observation process on potentially unobserved outcomes, the MMRM treatment
effect estimator remains consistent for the null value under the global
null hypothesis.

For MMRM-based hypothesis testing to be valid, two components need to be
considered. First, the treatment effect estimator should be centered at the null value
under the null hypothesis. Second, its variance estimator should
provide valid inference. The second component has been investigated by
Ishii et al.\cite{ishii_etal_24}, who showed that a robust variance estimator for MMRM
treatment effect estimators can remain consistent in the presence of missing
data, including MNAR missingness, and model misspecification. What remains
less clear is the first component: under what MNAR mechanisms does the MMRM
treatment effect estimator itself remain consistent for the null value?

We therefore investigate conditions under which the MMRM treatment effect estimator remains consistent for the null value under MNAR missingness. We derive a sufficient condition under the global null hypothesis and examine its interpretation in terms of the relationship between the longitudinal outcomes and the observation process. We further consider whether the resulting condition is compatible with plausible missing-data mechanisms.

Because the theoretical result is asymptotic, we complement it with two simulation studies to evaluate finite-sample performance. A fully synthetic simulation study examines the behavior of the MMRM method across a wide range of finite-sample settings, including settings not directly covered by the theoretical result, particularly nonmonotone missingness. A data-based bootstrap simulation using longitudinal data from an actual randomized clinical trial further evaluates finite-sample performance while preserving empirically observed outcome distributions, within-subject associations, and missing-data patterns.

\section{Consistency under a proportional observation condition}
\label{sec:consistency}

\subsection{General setting and proportional observation condition}

Consider a longitudinal study with \(T\) post-baseline time points and $K \geq 2$ groups. Let
\(\mathbf{Y}=(Y_1,\ldots,Y_T)^\top\) denote the full response vector, and let
$G \in \{1,\ldots,K\}$ denote the group indicator.
The MMRM method is based on the following normal marginal model:
\[
\mathbf{Y}\mid G = g
\sim
N(\mu^{(g)},\Sigma),
\]
where \(\mu^{(g)}=(\mu_1^{(g)},\ldots,\mu_T^{(g)})^\top\) is the
group-specific mean vector in group \(g\).
The covariance matrix
\(\Sigma\)
is assumed to be unknown, positive definite, and unstructured. 
The mean vector
\[
\mu
=
((\mu^{(1)})^\top, \ldots, (\mu^{(K)})^\top)^\top
\]
is estimated from the observed-data likelihood together with an estimator
of the within-subject covariance matrix.

Let \(y\) denote a realized value of the full response vector
\(\mathbf{Y}\), and let \(f(y\mid G=g)\) denote the true full-data density
in group \(g\). We consider the global null hypothesis

\[
H_0:
f(y\mid G=1)
= \cdots = 
f(y\mid G=K)
\quad
\text{for all }y.
\]
Thus, under \(H_0\), the full outcome distribution is identical across groups.

Under the global null hypothesis, treatment has no effect on the
longitudinal outcome distribution. In such a setting, it may be reasonable
to expect that the way in which the outcome itself affects the probability
of observation is similar across groups. This does not imply,
however, that the overall observation probabilities must be identical.
Treatment may affect dropout through other pathways, such as adverse events
or treatment burden, even when it has no effect on the outcome of interest.

Motivated by this distinction, we consider a condition that allows the
probability of observation to depend on the full outcome vector, including
potentially unobserved outcomes, while allowing group differences in the overall probability of observation. The essential requirement is
that the outcome-dependent component of the observation process is common
across groups up to a proportional factor.

Let \(\mathcal{R}_t\) denote the observation indicator for \(Y_t\), where \(\mathcal{R}_t=1\)
if \(Y_t\) is observed and \(\mathcal{R}_t=0\) otherwise. Taking group 1 as the
reference group, we define the proportional observation condition as follows. For each
time point \(t=1,\ldots,T\) and group $g = 2,\ldots,K$, there exists a constant
\(c_{g,t}>0\) such that
\[
\Pr(\mathcal{R}_t=1\mid\mathbf{Y}=y,G=g)
=
c_{g,t}
\Pr(\mathcal{R}_t=1\mid\mathbf{Y}=y,G=1)
\quad
\text{for all }y.
\]
This condition does not require MAR. The probability of observation may
depend on any component of the full response vector, including potentially
unobserved outcomes, and the missingness mechanism may therefore be MNAR.
Instead, the condition requires the ratio of the observation probability in
each group to that in the reference group to be constant with respect to the
full response vector.

\subsection{Main result for monotone missingness}

For the theoretical result, we focus on monotone missingness.
We restrict attention to subjects with at least one post-baseline
observation (i.e., \(\mathcal{R}_t = 1\) for some \(t\)). This restriction is
consistent with ICH E9 (Section 5.2.1), which identifies lack of any post-randomization data
as one of the limited circumstances that might lead to exclusion of randomized
subjects from the full analysis set (FAS), while emphasizing that such
exclusions should be justified \citep{ich_e9}.
Using the observation indicators \(\mathcal{R}_t\) defined above,
monotone missingness implies $\mathcal{R}_1 \geq \cdots \geq \mathcal{R}_T.$ 
The observation pattern can therefore be represented by
\[
R
=
\sum_{t=1}^T \mathcal{R}_t
\in
\{1,\ldots,T\}.
\]
Here, \(R\) represents the number of observed time points and, equivalently,
the last observed time point. Thus, \(Y_t\) is observed if and only if
\(R\geq t\).

\begin{theorem}
\label{thm:main}
Under the setting described above, suppose that the global null hypothesis
\(H_0\) and the proportional observation condition hold, that
\[
\mathrm{E}(\|\mathbf{Y}\|\mid G=g)<\infty,
\qquad g=1,\ldots,K,
\]
and that
\[
\Pr(R=T\mid G=g)>0,\qquad\Pr(G=g)>0,\qquad
g=1,\ldots,K.
\]
Let $\hat{\Sigma}$ denote the covariance estimator used in the MMRM
analysis, and suppose that
\[
\hat{\Sigma}\xrightarrow{p}\Sigma^*
\]
for some positive-definite matrix $\Sigma^*$. Then, the corresponding
estimators of the group-specific mean vectors obtained from the MMRM
observed-data likelihood satisfy
\[
\hat{\mu}^{(g)}(\hat{\Sigma}) - \hat{\mu}^{(g')}(\hat{\Sigma})
\xrightarrow{p}
\boldsymbol{0}
\]
for any $g, g' = 1,\ldots,K$ as the sample size in each group tends to infinity.
\end{theorem}

Notably, the result does not require the true outcome distribution to be normal, nor does it impose a parametric structure on the visit-specific means. In addition, \(\Sigma^*\) need not equal the true within-subject covariance matrix; the covariance estimator is required only to converge to a positive-definite limit. Thus, the result does not rely on correct specification of either the working normal distribution or the within-subject covariance structure. Theorem~\ref{thm:main} therefore addresses the centering component required for MMRM-based testing under the null hypothesis.
Valid inference additionally requires an appropriate variance estimator, as
discussed in Section~\ref{sec:discussion}.

The proof of Theorem~\ref{thm:main} is provided in Appendix~\ref{app:proof}.

\subsection{Intuition and examples}

\begin{figure}[!b]
    \centering
    \includegraphics[width=0.90\linewidth]{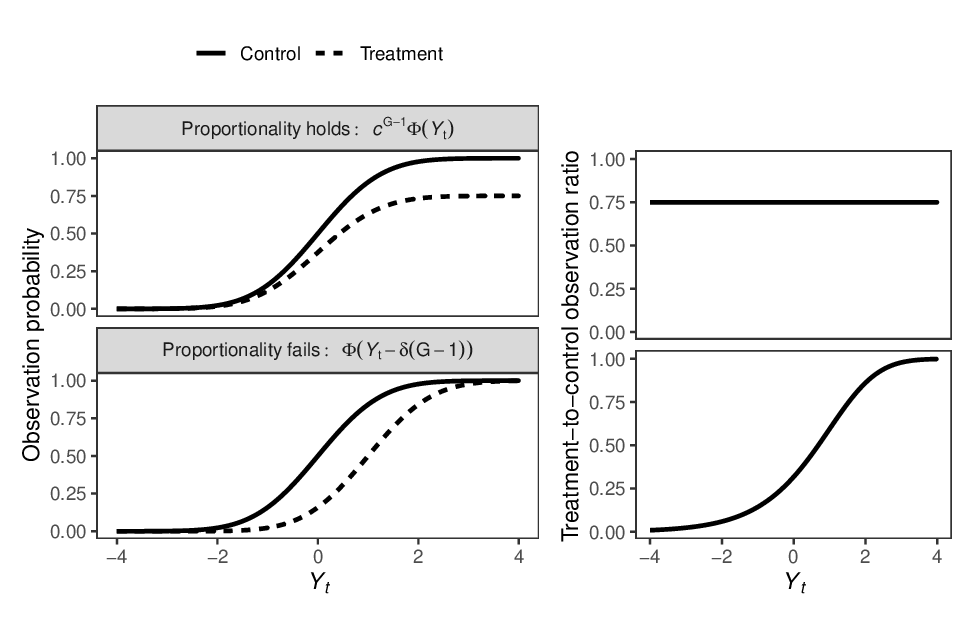}
    \caption{Illustration of proportional and non-proportional observation
    mechanisms. Left panels show observation probabilities, and right panels
    show the treatment-to-control observation ratio. Under proportionality
    (top), the ratio is constant with respect to the outcome; under a shifted
    probit model (bottom), the ratio varies with the outcome.}
     \label{fig:proportionality}
\end{figure}

To provide an intuitive interpretation of the proportional observation
condition, we specialize to the two-group case (\(K=2\)), with \(G=1\)
denoting the control group and \(G=2\) the treatment group. Consider a fixed
time point \(t\). Under monotone missingness, \(\mathcal{R}_t=1\) is
equivalent to \(R\geq t\), that is, remaining under observation through time
point \(t\). Suppose that larger values of \(Y_t\) indicate better disease
status and, as a simple illustration, that the probability of remaining under
observation through time point \(t\) because of an outcome-dependent process is
\[
\Pr(R\geq t\mid\mathbf{Y}=y,G)
=
\Phi(y_t),
\]
where \(\Phi(\cdot)\) denotes the cumulative distribution function of the
standard normal distribution. Subjects with poorer outcomes therefore have
a lower probability of remaining under observation. Because this probability
depends on the potentially unobserved value of \(Y_t\), the corresponding
missingness mechanism may be MNAR.

\medskip\noindent\textbf{Example 1: Proportional observation with additional outcome-independent dropout.}\par
\nopagebreak[4]

Suppose first that this outcome-dependent dropout process operates in the
same way in both groups. Under the global null hypothesis, treatment has no
effect on the longitudinal outcome distribution, so a common
outcome-dependent component of dropout may be a natural starting point.
This does not require the overall dropout rate to be the same between groups.
The treatment group may additionally have dropout through pathways unrelated
to the outcome of interest, such as adverse events or treatment burden.

Let \(c_t\), \(0<c_t\leq1\), denote the probability that a subject in the
treatment group remains under observation through time point \(t\) with
respect to this additional outcome-independent dropout process. Thus,
\(1-c_t\) is the cumulative probability of dropout through this additional
process by time point \(t\). In general, \(c_t\) may depend on \(t\), and
would typically decrease as follow-up proceeds.

If the outcome-dependent and additional outcome-independent dropout processes
act independently, the overall probability of remaining under observation
through time point \(t\) can be written as
\[
\Pr(R\geq t\mid\mathbf{Y}=y,G)
=
c_t^{G-1}\Phi(y_t),
\qquad G=1,2.
\]
Consequently,
\[
\frac{
\Pr(R\geq t\mid\mathbf{Y}=y,G=2)
}{
\Pr(R\geq t\mid\mathbf{Y}=y,G=1)
}
=
c_t,
\]
which does not depend on \(y_t\). Thus, treatment may change the overall
probability of remaining under observation, while the way in which the
outcome itself affects observation remains common between groups. This is
precisely the structure permitted by the proportional observation condition.
The upper panels of Figure~\ref{fig:proportionality} illustrate this setting:
the observation probabilities differ between groups, but their ratio is
constant with respect to the outcome. For simplicity, \(c_t\) is denoted by
\(c\) in the figure.

\medskip\noindent\textbf{Example 2: Non-proportional observation.}\par
\nopagebreak[4]

Now consider instead
\[
\Pr(R\geq t\mid\mathbf{Y}=y,G)
=
\Phi\{y_t-\delta(G-1)\},
\qquad
G=1,2,\quad \delta>0.
\]
Then,
\[
\frac{
\Pr(R\geq t\mid\mathbf{Y}=y,G=2)
}{
\Pr(R\geq t\mid\mathbf{Y}=y,G=1)
}
=
\frac{\Phi(y_t-\delta)}{\Phi(y_t)},
\]
which depends on \(y_t\). The proportional observation condition is therefore
violated. In particular, the relative reduction in the probability of
remaining under observation is greater for subjects with poorer outcomes and
becomes smaller as the outcome improves.

Unlike Example 1, this mechanism assumes that, although treatment does not
affect the longitudinal outcome distribution under the global null hypothesis,
the effect of the outcome on the probability of remaining observed differs
between groups. Although mathematically well defined, this structure may be
less natural than that in Example 1 in many clinical settings.

The lower panels of Figure~\ref{fig:proportionality} illustrate this
non-proportional mechanism. Although the treatment group has a lower
observation probability throughout, the treatment-to-control observation
ratio varies with the outcome. Taken together, the two examples highlight
the essential meaning of the proportional observation condition in the
two-group case: it permits treatment-group differences in the overall
probability of observation, but requires the outcome dependence of the
observation process to be common between groups up to a time-specific
proportional factor.

\section{Simulation studies}
\subsection{Fully synthetic simulation study}
\medskip\noindent\textbf{Simulation design.}\par
\nopagebreak[4]
Simulation studies were conducted to evaluate the finite-sample performance
of the MMRM method under the global null hypothesis and to examine whether
the theoretical finding might extend empirically to nonmonotone missingness.

Longitudinal outcomes at baseline and five post-baseline time points were
generated from either a multivariate normal distribution or a multivariate
lognormal distribution. The mean and standard deviation vectors on the
original scale were specified as
$(100,90,80,70,75,80)^\top$ and $(25,20,15,10,15,20)^\top,$
respectively, where the first element corresponds to baseline. An
autoregressive correlation structure of order 1 with correlation parameter
\(0.8\) was used. In the lognormal scenarios, the correlation structure was
specified on the log-transformed scale. The outcome distributions were
identical between the treatment groups in all scenarios, corresponding to
the global null hypothesis.

Efficacy-related missingness was generated using an outcome-dependent probit
model. Let \(z_t\) denote the standardized outcome at time point \(t\); in
the lognormal scenarios, standardization was performed after log
transformation. The probability of efficacy-related missingness at
post-baseline time point \(t\) was specified as
\[
\Pr(M_t=1\mid\mathbf{Y})
=
\Phi(\gamma_0+0.5z_{t-1}+z_t),
\]
where \(M_t=1\) indicates that the outcome at time point \(t\) is missing
for an efficacy-related reason, and \(\Phi(\cdot)\) denotes the cumulative
distribution function of the standard normal distribution. The intercept \(\gamma_0\) was calibrated so that the final-visit
efficacy-related missingness proportion, after applying the dropout process,
was \(20\%\) or \(40\%\). We also
considered a setting with no efficacy-related missingness.

Among subjects with at least one efficacy-related missing outcome, dropout
was imposed with probability \(p_d\). For subjects selected to drop out, all
outcomes from the first efficacy-related missing time point onward were set
to missing. We considered \(p_d=0.4\) and \(p_d=1\). When \(p_d=0.4\),
efficacy-related missingness consisted of a mixture of intermittent
missingness and dropout and could therefore be nonmonotone. When \(p_d=1\),
all efficacy-related missingness resulted in dropout, producing monotone
missingness. When the efficacy-related missingness proportion was \(0\%\),
the value of \(p_d\) was irrelevant, and this setting was counted only once.

In addition, treatment-group-specific dropout due to adverse events was
generated independently of the outcome. The probability of
adverse-event-related dropout was set to \(p_a=0\) or \(p_a=0.2\) in the
treatment group and to zero in the control group. For subjects experiencing
adverse-event-related dropout, the dropout time point was selected randomly from the five post-baseline visits, and the outcome at that time point and all subsequent outcomes were set to missing. The final missing-data pattern was defined by combining
the efficacy-related and adverse-event-related missingness processes.

The total sample size was set to \(50\), \(100\), or \(200\). Control-to-treatment allocation ratios of \(1{:}1\) and approximately \(1{:}2\) were
considered. Under the \(1{:}1\) allocation, the control and treatment group
sample sizes were \((25,25)\), \((50,50)\), and \((100,100)\), respectively.
Under the \(1{:}2\) allocation, the corresponding group sample sizes were
\((17,33)\), \((33,67)\), and \((67,133)\).

The simulation factors therefore consisted of two outcome distributions,
three total sample sizes, two allocation ratios, five efficacy-related
missingness settings, and two adverse-event-related dropout probabilities.
The five efficacy-related missingness settings were
$0\%$,
$20\%\text{ with }p_d=0.4$,
$20\%\text{ with }p_d=1$, 
$40\%\text{ with }p_d=0.4$, 
and $40\%\text{ with }p_d=1$.
Thus, a total of
$
2\times3\times2\times5\times2=120
$
simulation conditions were examined.

Each simulation condition was replicated 100,000 times. For a true rejection
probability of 0.05, the Monte Carlo standard error is approximately 0.0007,
corresponding to an approximate pointwise 95\% Monte Carlo interval of
0.0486--0.0514 for the empirical test size in any individual condition.
Because 120 simulation conditions were examined, this pointwise interval was
not intended as a simultaneous criterion for assessing type I error control
across all conditions.
To reflect a commonly used implementation of MMRM for longitudinal continuous outcomes in randomized clinical trials, we used restricted maximum likelihood estimation, an unstructured within-subject covariance matrix, and Kenward--Roger inference \citep{kenward_roger_97}. The model included treatment group, visit, and their interaction as fixed effects, together with baseline outcome and its interaction with visit. Visit was treated as a categorical variable. The adjusted treatment difference at the final time point was the parameter of interest. Replicates in which the MMRM did not converge were excluded from the calculation of standardized bias and empirical test size, and the frequency of nonconvergence was recorded.

Finite-sample performance was evaluated in terms of standardized bias and
empirical test size. Because the true treatment difference was zero in all
scenarios, the standardized bias was calculated as
$\operatorname{mean}(\hat{\Delta})/\sigma_{e,T},$
where \(\hat{\Delta}\) denotes the estimated adjusted treatment difference
at the final time point and \(\sigma_{e,T}\) denotes the true residual
standard deviation at that time point after linear adjustment for baseline. The empirical test size was defined as the proportion
of simulated datasets in which the null hypothesis of no treatment effect
at the final time point was rejected using a two-sided test at a significance
level of \(0.05\).

All simulations were performed using R version 4.6.1
(R Core Team, Vienna, Austria) and the \texttt{mmrm} package,
version 0.3.18 \citep{mmrm}.

\medskip\noindent\textbf{Simulation results.}\par
\nopagebreak[4]
Across the 120 simulation conditions, the standardized bias ranged from
\(-0.0020\) to \(0.0019\), indicating negligible bias in all settings. The
empirical test size ranged from \(0.0470\) to \(0.0517\) and remained close to
the nominal level of \(0.05\). Nineteen of the 120 empirical test sizes fell
outside the pointwise Monte Carlo interval of 0.0486--0.0514; 18 were below
the interval and only one was above it. Values below the interval occurred
primarily in the lognormal scenarios with smaller sample sizes, suggesting
mild finite-sample conservatism rather than systematic type I error inflation.
These findings were otherwise consistent across outcome distributions, sample
sizes, allocation ratios, efficacy-related missingness proportions, missingness
patterns, and adverse-event-related dropout probabilities. Nonconvergence was
rare: it occurred in 6 of the 120 conditions, and the largest nonconvergence
rate was \(0.012\%\) (12 of 100,000 replicates). Detailed condition-specific
results, including nonconvergence rates, are provided in Supporting Information
Tables~S1 and~S2.

\subsection{Bootstrap simulation based on clinical trial data}
To complement the fully synthetic simulation study, we conducted a
bootstrap simulation based on longitudinal data from an actual randomized
clinical trial. 

\medskip\noindent\textbf{Data source.}\par
\nopagebreak[4]
We used the \texttt{qolef} dataset distributed with the R package
\texttt{isni}, version 1.3 \citep{Xie2018}.
The data originated from a quality-of-life companion study of Southwest
Oncology Group trial INT-0105, a randomized double-blind phase III trial
in patients with metastatic prostate cancer. The parent trial compared
bilateral orchiectomy plus flutamide with bilateral orchiectomy plus placebo
to evaluate whether the addition of flutamide provided a clinically
meaningful survival benefit \citep{Eisenberger1998}. Because treatment for
metastatic prostate cancer was primarily palliative, the companion study
was conducted to assess the achievement of palliation and to characterize
the beneficial and adverse effects of the two androgen-deprivation
strategies on quality of life \citep{Moinpour1998}.

Quality-of-life outcomes were assessed at randomization and at 1, 3, and
6 months thereafter. The prespecified quality-of-life domains included
three treatment-specific measures---diarrhea, gas pain, and body
image---as well as physical functioning and emotional functioning.
The present bootstrap simulation used the emotional functioning (EF) score,
for which higher values indicate better emotional functioning. The
\texttt{qolef} dataset contained 715 subjects with an observed baseline EF
score. For each subject, the dataset included the baseline EF score, EF scores
at the three post-baseline visits, the original treatment assignment, and
indicators distinguishing observed measurements, intermittent missingness,
and dropout.
Among the 715 subjects, 497 (69.5\%) had complete EF measurements at all
three post-baseline visits. Intermittent missingness occurred in
68 subjects (9.5\%), and dropout occurred in 158 subjects (22.1\%),
including eight subjects who experienced both intermittent missingness
and subsequent dropout. Thus, the dataset contained both monotone and
nonmonotone missing-data patterns.

\medskip\noindent\textbf{Simulation design.}\par
\nopagebreak[4]
The original treatment-group labels were discarded, and all subjects were
pooled to form a common empirical population. For each simulation replicate,
\(N\) subjects were sampled with replacement from the pooled dataset. The
baseline EF score, post-baseline EF measurements, and missing-data pattern of
each sampled subject were retained together. Repeated selections of the same
source subject were treated as distinct subjects and assigned unique subject
identifiers.

The resampled subjects were then randomly assigned to two treatment groups in
a \(1{:}1\) ratio, independently of their outcomes and original missing-data
patterns. Thus, the newly assigned treatment had no effect on the longitudinal
outcomes or on the empirically observed missing-data process. To additionally
examine a treatment-group difference in observation probabilities of the type
permitted by the proportional observation condition, we superimposed
additional outcome-independent dropout in the treatment group only. This
additional dropout, intended to represent a pathway such as adverse-event-related 
discontinuation, occurred with probability \(p_a=0\) or \(0.2\);
the corresponding probability was zero in the control group. For subjects
selected for additional dropout, the dropout visit was chosen at random from
the 1-, 3-, and 6-month visits, and the outcome at that visit and all
subsequent outcomes were set to missing. The original empirically observed
missing-data patterns were otherwise retained. Total sample sizes of
\(N=50\), \(100\), and \(200\) were considered.

Each simulation condition was replicated 100,000 times.
The same MMRM analysis as in Section~3.1 was applied to the EF scores measured
at 1, 3, and 6 months, with baseline EF used as the baseline covariate. The
adjusted treatment-group difference at 6 months was the parameter of interest.
For this data-based simulation, standardized bias was scaled by the standard
deviation of the observed 6-month EF scores in the pooled source dataset.

This bootstrap simulation preserved the empirically observed outcome
distribution, within-subject association, and mixture of dropout and
intermittent missingness, while allowing an additional treatment-specific,
outcome-independent dropout process. However, because the unobserved outcomes
were not available, it should be interpreted as a data-based evaluation under
realistic observed missing-data patterns rather than as a simulation under a
fully specified MNAR mechanism.

\medskip\noindent\textbf{Simulation results.}\par
\nopagebreak[4]
When no additional treatment-specific dropout was imposed
(\(p_a=0\)), the standardized biases for total sample sizes of \(50\),
\(100\), and \(200\) were \(-0.0001\), \(-0.0010\), and \(0.0002\),
respectively, and the corresponding empirical test sizes were \(0.0482\),
\(0.0490\), and \(0.0497\). With additional outcome-independent dropout in
the treatment group (\(p_a=0.2\)), the standardized biases were \(0.0005\),
\(-0.0008\), and \(0.0001\), and the corresponding empirical test sizes were
\(0.0479\), \(0.0489\), and \(0.0496\), respectively. No model-fitting
failures occurred in any of the bootstrap replicates.

Thus, the estimated treatment-group differences remained essentially centered
at zero and no type I error inflation was observed, even when the treatment
group had additional outcome-independent dropout. These findings were
consistent with the fully synthetic simulation study and supported the
finite-sample performance of MMRM-based testing under observation mechanisms
compatible with the proportional observation condition.

\section{Discussion}
\label{sec:discussion}
In this study, we examined the use of the MMRM method under MNAR missingness
from the perspective of hypothesis testing. We showed that, under the global
null hypothesis, monotone missingness, and the proportional observation
condition, MMRM estimators of between-group mean differences converge to
zero for an arbitrary number of post-baseline visits, even when the
probability of observation depends on potentially unobserved outcomes. The
simulation results supported this theoretical finding in broader finite-sample
settings, including baseline adjustment, nonmonotone missingness, nonnormal
outcomes, and additional
treatment-group-specific dropout unrelated to the outcome. A data-based
bootstrap simulation using longitudinal quality-of-life data from an actual
randomized clinical trial further showed negligible bias and appropriate test
size while preserving the empirically observed outcome distribution,
within-subject association, and mixture of monotone and nonmonotone
missingness.

The present theoretical result complements previous work on robust variance
estimation for the MMRM method. Ishii et al.\cite{ishii_etal_24} showed that a robust
variance estimator for MMRM treatment effect estimators can remain consistent
under model misspecification and missing data, including MNAR missingness.
Consistency of the variance estimator alone, however, is not sufficient for
valid hypothesis testing; the treatment effect estimator itself must also
converge to the null value under the null hypothesis. The present study
addresses this remaining component by deriving a condition under which the
MMRM treatment effect estimator is consistent under the global null
hypothesis. Taken together, these results provide a basis for understanding
why MMRM-based tests may exhibit appropriate type I error control under some
MNAR mechanisms.

The proportional observation condition is not as restrictive as it may first
appear. It allows the overall probability of observation to differ between
treatment groups, while requiring that the outcome-dependent component of the
observation process is common between groups up to a proportional factor.
For example, in a two-group trial, patients with worse outcomes may be more
likely to discontinue follow-up in both groups, while the treatment group may
additionally have outcome-independent dropout due to adverse events or
treatment burden. In
such a setting, observation probabilities can differ between groups, but the
way in which the outcome affects observation remains the same. This type of
mechanism is clinically plausible under the global null hypothesis, where
treatment has no effect on the outcome distribution. By contrast, violation
of the proportional observation condition means that the relationship between
the outcome and observation differs between treatment groups. Although such a
mechanism is mathematically well defined, it is less straightforward to
motivate clinically under the global null hypothesis because it requires the
outcome dependence of the observation process itself to differ between groups.
Thus, the proportional observation condition includes a practically plausible
class of MNAR mechanisms while still allowing treatment-group differences in
overall observation probabilities.

The focus on the global null hypothesis also warrants consideration. The
global null hypothesis represents a situation in which treatment has no
effect on the full longitudinal outcome distribution and therefore corresponds
to a scientifically meaningful no-treatment-effect configuration, rather than
merely a technical assumption introduced for mathematical convenience. In
contrast, a pointwise null hypothesis at the target visit may hold even when
treatment affects outcomes at other visits. Such a configuration is a null
hypothesis for the target contrast but is not a no-treatment-effect
configuration for the longitudinal outcome process as a whole. This distinction
does not remove the need to control type I error under the prespecified
pointwise null hypothesis, and the present result does not cover all such null
configurations. Nevertheless, the global null hypothesis provides a natural
starting point for investigating MNAR mechanisms because it allows us to ask
whether outcome-dependent observation alone can generate spurious evidence of
a treatment effect when treatment itself has no effect on the outcome process.

These results suggest that MNAR missingness does not automatically invalidate
MMRM-based testing. In practice, outcome-dependent dropout is often regarded
as a reason to avoid the MMRM method or to replace it with more conservative
analyses. Such approaches may be appropriate when the primary objective is
estimation under a particular estimand, especially when the interpretation of
the MMRM estimate is questionable. However, overly conservative primary
analyses can have non-negligible consequences for clinical development. Drug
development remains a resource-intensive process, characterized by high
attrition, long development timelines, and substantial research and
development costs
\citep{wong_etal_19,wouters_etal_20,sertkaya_etal_24}.
Recent analyses also suggest that R\&D intensity and overall R\&D spending
remain high, underscoring the continuing pressure to improve the efficiency
of clinical development \citep{sertkaya_etal_24}. Therefore, even a modest
loss of power at the confirmatory-trial stage can affect the required sample
size, the probability of trial success, and ultimately the efficiency of
development programs. From this perspective, clarifying when MMRM-based
testing retains type I error control under MNAR mechanisms is not merely a
technical issue; it may provide one way to avoid unnecessary conservatism
while preserving the validity of statistical hypothesis testing in
confirmatory clinical trials.

Ideally, hypothesis testing and estimation should be aligned: the test should
support the same treatment effect that is summarized by the reported estimate,
and both should be interpretable under the prespecified estimand. This
alignment is particularly important under the ICH E9(R1) estimand framework.
Nevertheless, statistical hypothesis testing in a confirmatory clinical trial
has a more limited role. Its primary purpose is to assess whether the data
provide sufficient evidence to reject the null hypothesis of no treatment
effect while controlling the type I error rate. It does not, by itself,
determine the magnitude or clinical interpretation of the treatment effect
under a particular estimand. These latter issues belong primarily to
estimation and sensitivity analyses.

This distinction is not unique to the missing-data setting. In group
sequential trials, for example, hypothesis testing is conducted using
prespecified sequential boundaries to control the type I error rate, whereas
conventional fixed-sample estimates of the treatment effect may be biased
after interim analyses \citep{whitehead_86}. Accordingly, FDA guidance on
adaptive designs recommends that methods for calculating estimates and
confidence intervals that appropriately account for the group sequential
design should be prospectively planned and used for reporting results
\citep{fda_adaptive_19}. This example illustrates that testing and estimation
can, when necessary, be handled using different statistical tools while
remaining part of a coherent analysis strategy. In a similar sense, the MMRM
method may be considered as a testing procedure under certain MNAR mechanisms,
provided that the interpretation of the estimated treatment effect is
addressed separately through estimand-based analyses and sensitivity analyses.

From this perspective, requiring the primary hypothesis test itself to
resolve all questions about the interpretation of treatment effects under
MNAR missingness may be unnecessarily restrictive. When the objective is to
make a confirmatory claim that the treatment effect is not zero while
controlling the type I error rate, the key question for the test is whether
type I error is preserved under the relevant null hypothesis. The present
study addresses this question by identifying conditions under which
MMRM-based testing remains valid under MNAR mechanisms. This does not remove
the need for estimand-based interpretation of treatment effect estimates,
but it separates the validity of the confirmatory claim from the
interpretation of the reported estimate. This distinction may be particularly
relevant in trials where patient recruitment is difficult, sample size
expansion is limited, or repeated follow-up imposes a substantial burden on
participants. In such settings, avoiding unnecessary conservatism in the
primary test, while retaining appropriate sensitivity analyses for
estimation, may be an important practical consideration.

Our findings should not be interpreted as a general justification for the
MMRM method under MNAR missingness. The theoretical result concerns the
global null hypothesis and therefore addresses the validity of hypothesis
testing, not unbiased estimation under alternatives. When a true treatment
effect exists, outcome-dependent missingness may induce bias in the MMRM
estimate, and the resulting estimate may not correspond to the treatment
effect targeted by the estimand. Thus, for estimation, the assumptions
underlying the analysis must still be carefully assessed, and sensitivity
analyses based on alternative missing-data assumptions remain essential
\citep{ich_e9r1,little_rubin,molenberghs_kenward}.

Several limitations should be noted. First, the theoretical result assumes
monotone missingness. Although both simulation studies included nonmonotone
missing-data patterns, a formal extension to general missing-data patterns
remains future work. Second, the proportional observation condition is a
sufficient condition, not a necessary condition. Other observation mechanisms
may also yield valid MMRM-based testing under the global null hypothesis, but
they were not characterized in this study. Third, the present study focused
on the global null hypothesis. Local null hypotheses, in which treatment
effects may be present at some visits but absent at the target visit, were
not examined. Relatedly, the simulations used restricted maximum likelihood
estimation with the Kenward--Roger method under a covariance structure common
to the treatment groups, and their results should not be extrapolated to
settings with group-specific covariance structures. In such settings,
nonorthogonality between mean and covariance parameters may affect standard
error estimation, particularly in small samples, and procedures that do not
rely on orthogonality may be preferable
\citep{maruo_etal_20,maruo_etal_26}. Finally, this study does not determine
how large a departure from the proportional observation condition can be
tolerated while maintaining acceptable type I error control in practice.

In conclusion, MNAR missingness does not necessarily preclude the use of the
MMRM method as a hypothesis-testing procedure. Under the global null
hypothesis, if the proportional observation condition holds, MMRM estimators
of between-group mean differences remain consistent for zero despite
outcome-dependent missingness. Together with the simulation results, this
finding suggests that MMRM-based testing may retain validity under practically
relevant MNAR mechanisms. This conclusion is deliberately limited to
hypothesis testing; treatment effect estimation under alternatives continues
to require estimand-based interpretation and appropriate sensitivity
analyses.

\section*{Acknowledgments}
The authors used ChatGPT (OpenAI) for English language refinement and
assistance with R code development and correction. All AI-assisted output was
reviewed and revised by the authors, who take full responsibility for the
content of the manuscript.

\section*{Author Contributions}
Kazushi Maruo and Ryota Ishii contributed equally to the conceptualization of the study, mathematical investigation, simulation studies, and drafting of the manuscript. Yusuke Yamaguchi and Naoki Isogawa supported the conceptualization of the study and critically reviewed and revised the manuscript. Keisuke Hanada supported the mathematical investigation and critically reviewed and revised the manuscript. Masahiko Gosho supervised the study and critically reviewed and revised the manuscript.

\section*{Funding}
This work was supported by JSPS KAKENHI Grant Numbers 23K11003 and 26K02873.

\section*{Conflicts of Interest}
The authors declare no conflicts of interest.

\section*{Data Availability Statement}
The data used in this paper to illustrate our findings in the data-based
bootstrap simulation are available in the \texttt{qolef} dataset distributed
with the R package \texttt{isni} and are described by \citet{Xie2018}. No
external data were used for the fully synthetic simulation study.

\bibliographystyle{plainnat}
\bibliography{bibfile}

\clearpage
\appendix
\section{Proof of Theorem 1}
\label{app:proof}

For $t = 1,\ldots,T$ and $g = 1,\ldots,K$, let $p_t^{(g)} = \Pr(R = t, G = g)$.
For any positive-definite matrix $\Omega = [\omega_{j,k}]_{1 \leq j,k \leq T}$, we first consider the MMRM method with a fixed variance-covariance matrix $\Omega$.
Let $\Omega_t$ denote the $t \times t$ submatrix of $\Omega$ corresponding to the first $t$ time points.
Let $O_{k,k'}$ be the $k \times k'$ zero matrix and $e_{n,k}$ be the $k$th standard basis vector in $\mathbb{R}^n$.
Define the $T \times T$ matrices $C_1(\Omega), \ldots C_T(\Omega)$ and $C^{(g)}(\Omega)$ for $g = 1,\ldots, K$, as follows:
\begin{align*}
C_t(\Omega) = \begin{bmatrix}
\Omega_t^{-1} &  \\
 & O_{T - t, T - t}
\end{bmatrix}, \quad C^{(g)}(\Omega) = \sum_{t = 1}^T p_t^{(g)}C_t(\Omega).
\end{align*}
The condition $\Pr(R = T \mid G = g) > 0$ in Theorem 1 implies that $p_T^{(g)} > 0$ and hence that $C^{(g)}$ is invertible. 
\begin{proposition} \label{prop:MLE1}
For each group $g = 1,\ldots,K$, the probability limit of $\hat{\mu}^{(g)}(\Omega)$ is given by
\[
\sum_{t = 1}^T(p_t^{(g)} + \cdots + p_T^{(g)})(C^{(g)})^{-1}(C_t - C_{t - 1})\nu_t^{(g)},
\]
where $C_0 = O_{T, T}$, $C_t = C_t(\Omega)$, $C^{(g)} = C^{(g)}(\Omega)$, and $\nu_t^{(g)} = \mathrm{E}[\mathbf{Y} \mid R \ge t, G = g]$ for $t = 1,\ldots,T$.
\end{proposition}
\begin{proof}
For response pattern $R = t$ and group $G = g$, the design matrix $X_t^{(g)}$ is given by $X_t^{(g)} = e_{K,g}^\top \otimes \begin{bmatrix} E_t & O_{t, T - t} \end{bmatrix}$, where $E_k$ is the $k \times k$ identity matrix.
Let $m_t^{(g)} = \mathrm{E}[\mathbf{Y} \mid R = t, G = g]$.
The expected score is $Z - C\mu$, where 
\begin{align*}
C &= \sum_{t = 1}^T \sum_{g = 1}^K p_t^{(g)}(X_t^{(g)})^\top \Omega_t^{-1}X_t^{(g)} = \begin{bmatrix}
C^{(1)} & & \\
& \ddots & \\
& & C^{(K)}
\end{bmatrix}, \\
Z &= \sum_{t = 1}^T \sum_{g = 1}^K(X_t^{(g)})^\top \Omega_t^{-1}X_t^{(g)}(e_{K,g}\otimes p_t^{(g)}m_t^{(g)}) = \sum_{t = 1}^T \begin{bmatrix} p_t^{(1)}C_tm_t^{(1)} \\
\vdots \\
p_t^{(K)}C_tm_t^{(K)} 
\end{bmatrix}.
\end{align*}
Therefore, $\hat{\mu}$ converges in probability to $C^{-1}Z$, which can be written as
\[
C^{-1}Z =  \begin{bmatrix} \sum_{t = 1}^Tp_t^{(1)}(C^{(1)})^{-1}C_tm_t^{(1)} \\
\vdots \\
\sum_{t = 1}^Tp_t^{(K)}(C^{(K)})^{-1}C_tm_t^{(K)} 
\end{bmatrix}.
\]
Hence, the probability limit of $\hat{\mu}^{(g)}$ can be expressed as
\[
\sum_{t = 1}^Tp_t^{(g)}(C^{(g)})^{-1}C_tm_t^{(g)} = \sum_{t = 1}^T(C^{(g)})^{-1}(C_t - C_{t - 1})(p_t^{(g)}m_t^{(g)} + \cdots + p_T^{(g)}m_T^{(g)}).
\]
The identity
\begin{align*}
(p_t^{(g)} + \cdots + p_T^{(g)})\nu_t^{(g)} &= p_t^{(g)}\mathrm{E}[\mathbf{Y} \mid R = t, G = g] + \cdots + p_T^{(g)}\mathrm{E}[\mathbf{Y} \mid R = T, G = g] \\
&=p_t^{(g)}m_t^{(g)} + \cdots + p_T^{(g)}m_T^{(g)}
\end{align*}
yields the desired result.
\end{proof}

The following lemma gives a convenient expression for the matrix factor $(C^{(g)})^{-1}(C_t - C_{t - 1})$ appearing in Proposition \ref{prop:MLE1}.
\begin{biomlemma} \label{lem:Cginv}
$(C^{(g)})^{-1}(C_t - C_{t - 1}) = (p_t^{(g)} + \cdots + p_T^{(g)})^{-1}\Omega(C_t - C_{t - 1})$.
\end{biomlemma}
\begin{proof}
Let $\boldsymbol{0}_n$ denote the $n$-dimensional zero vector, and define $u_t$, $t = 1,\ldots,T$, as follows:
\[
u_1 = \begin{bmatrix}
1 \\
\boldsymbol{0}_{T - 1}
\end{bmatrix} \text{ and } u_t = \begin{bmatrix}
-\Omega_{t - 1}^{-1}\begin{bmatrix}
\omega_{1,t} \\
\vdots \\
\omega_{t - 1,t}
\end{bmatrix}\\
1 \\
\boldsymbol{0}_{T - t}
\end{bmatrix} \text{ for $t = 2,\ldots, T$}.
\]
Let $s_t$ denote the $(t,t)$ element of $\Omega_t^{-1}$, and define $U$, $D$, and $D^{(g)}$ as follows:
\begin{gather*}
U = \begin{bmatrix}
u_1 & \cdots & u_T
\end{bmatrix}, \quad D = \mathrm{diag}\left(s_1,\ldots,s_T\right), \\
D^{(g)} = \mathrm{diag}\left((p_1^{(g)} + \cdots + p_T^{(g)})s_1,\ldots,(p_t^{(g)} + \cdots + p_T^{(g)})s_t,\ldots,p_T^{(g)}s_T\right).
\end{gather*}
Applying the block matrix inversion formula to $\Omega_t$ gives $C_{t} - C_{t - 1} = s_tu_tu_t^\top = Us_te_{T,t}e_{T,t}^\top U^\top$ for $t = 1,\ldots,T$.
Hence, $\Omega^{-1}$ and $C^{(g)}$ can be expressed as follows:
\begin{align*}
\Omega^{-1} &= C_T = \sum_{t = 1}^T (C_t - C_{t - 1}) = \sum_{t = 1}^T Us_te_{T,t}e_{T,t}^\top U^\top = UD U^\top \\
C^{(g)} &= \sum_{t = 1}^T (p_t^{(g)} + \cdots + p_T^{(g)})(C_t - C_{t - 1}) = \sum_{t = 1}^T U(p_t^{(g)} + \cdots + p_T^{(g)})s_te_{T,t}e_{T,t}^\top U^\top =  UD^{(g)}U^\top.
\end{align*}
Therefore,
\begin{align*}
(C^{(g)})^{-1}(C_t - C_{t - 1}) &= (U^\top)^{-1}(D^{(g)})^{-1}U^{-1}(Us_te_{T,t}e_{T,t}^\top U^\top) \\
&= (p_t^{(g)} + \cdots + p_T^{(g)})^{-1}(U^\top)^{-1}e_{T,t}e_{T,t}^\top U^\top \\
&= (p_t^{(g)} + \cdots + p_T^{(g)})^{-1}(U^\top)^{-1}D^{-1}U^{-1}(s_tUe_{T,t}e_{T,t}^\top U^\top) \\
&= (p_t^{(g)} + \cdots + p_T^{(g)})^{-1}\Omega(C_t - C_{t - 1}).
\end{align*}
\end{proof}

Combining Proposition \ref{prop:MLE1} with Lemma \ref{lem:Cginv} shows that the probability limit of $\hat{\mu}^{(g)}(\Omega)$ is given by
\[
\sum_{t = 1}^T\Omega(C_t(\Omega) - C_{t - 1}(\Omega))\nu_t^{(g)}.
\]
Because the first moments of $\mathbf{Y}$ are finite and $\Sigma^*$ is positive definite, we may choose a compact neighborhood of $\Sigma^*$ consisting only of positive-definite matrices. On this neighborhood, the empirical score matrices and vectors converge uniformly to their expectations, and the corresponding inverse matrices remain well defined with probability tending to one. Hence, $\hat{\mu}(\Omega)$ converges uniformly in probability to its probability limit on this neighborhood. Since $\hat{\Sigma}\xrightarrow{p}\Sigma^*$, the continuous mapping theorem then yields the following result.

\begin{proposition} \label{prop:MLE2}
The probability limit of $\hat{\mu}^{(g)}(\hat{\Sigma})$ is given by
\[
\sum_{t = 1}^T\Sigma^*(C_t(\Sigma^*) - C_{t - 1}(\Sigma^*))\nu_t^{(g)}.
\]
\end{proposition}

Propositions \ref{prop:MLE1} and \ref{prop:MLE2} do not require either the null hypothesis $H_0$ or the proportional observation condition.
For any $g = 2,\ldots,K$, under $H_0$ and the proportional observation condition, we have
\begin{align*}
f(y \mid R \geq t, G = g) &= \frac{\Pr(R \geq t \mid y, G = g)f(y \mid G = g)}{\int \Pr(R \geq t \mid y, G = g)f(y \mid G = g)dy} \\
&= \frac{c_{g,t}\Pr(R \geq t \mid y, G = 1)f(y \mid G = 1)}{c_{g,t}\int \Pr(R \geq t \mid y, G = 1)f(y \mid G = 1)dy} \\
&= f(y \mid R \geq t, G = 1)
\end{align*}
which implies $\nu_t^{(1)} = \cdots = \nu_t^{(K)}$.
Therefore, Theorem 1 follows immediately from Proposition \ref{prop:MLE2} and this identity.

\clearpage
\section{Detailed simulation results}

\renewcommand{\thetable}{S\arabic{table}}
\setcounter{table}{0}

The fully synthetic simulation study comprised 120 conditions. In the tables
below, the efficacy-related missingness proportion, \(p_d\), and the adverse-event
dropout probability are all reported on the 0--1 probability scale. The quantity
\(p_d\) denotes the probability that an efficacy-related missing outcome resulted
in dropout and is not applicable when the efficacy-related missingness proportion
was 0. ``AE dropout'' denotes the treatment-group adverse-event dropout probability;
it was 0 in the control group. The allocation ratio is
control:treatment. Performance measures were calculated among converged fits.

\scriptsize
\setlength{\tabcolsep}{3.0pt}
\renewcommand{\arraystretch}{1.05}
\begin{longtable}{cccccccc}
\caption{Detailed simulation results for the normal outcome distribution. Nonconvergence is reported as a percentage of the 100,000 simulation replicates.}\label{tab:supp-normal}\\
\toprule
\makecell{Efficacy-related\\missingness proportion} & \makecell{Dropout\\probability $p_d$} & \makecell{AE dropout\\probability} & \makecell{Allocation\\ratio} & $N$ & \makecell{Standardized\\bias} & \makecell{Test\\size} & \makecell{Nonconvergence\\(\%)} \\
\midrule
\endfirsthead
\multicolumn{8}{c}{Supplementary Table S1 (continued)}\\
\toprule
\makecell{Efficacy-related\\missingness proportion} & \makecell{Dropout\\probability $p_d$} & \makecell{AE dropout\\probability} & \makecell{Allocation\\ratio} & $N$ & \makecell{Standardized\\bias} & \makecell{Test\\size} & \makecell{Nonconvergence\\(\%)} \\
\midrule
\endhead
\midrule
\multicolumn{8}{r}{Continued on next page}\\
\endfoot
\bottomrule
\endlastfoot
0.0 & -- & 0.0 & 1:1 & 50 & 0.0005 & 0.0505 & 0.000 \\
0.0 & -- & 0.0 & 1:1 & 100 & 0.0004 & 0.0487 & 0.000 \\
0.0 & -- & 0.0 & 1:1 & 200 & 0.0003 & 0.0496 & 0.000 \\
0.0 & -- & 0.0 & 1:2 & 50 & 0.0000 & 0.0498 & 0.000 \\
0.0 & -- & 0.0 & 1:2 & 100 & -0.0004 & 0.0492 & 0.000 \\
0.0 & -- & 0.0 & 1:2 & 200 & -0.0003 & 0.0488 & 0.000 \\
0.0 & -- & 0.2 & 1:1 & 50 & 0.0019 & 0.0495 & 0.000 \\
0.0 & -- & 0.2 & 1:1 & 100 & -0.0001 & 0.0508 & 0.000 \\
0.0 & -- & 0.2 & 1:1 & 200 & 0.0002 & 0.0501 & 0.000 \\
0.0 & -- & 0.2 & 1:2 & 50 & 0.0004 & 0.0501 & 0.000 \\
0.0 & -- & 0.2 & 1:2 & 100 & -0.0004 & 0.0501 & 0.000 \\
0.0 & -- & 0.2 & 1:2 & 200 & -0.0005 & 0.0499 & 0.000 \\
0.2 & 0.4 & 0.0 & 1:1 & 50 & 0.0005 & 0.0502 & 0.000 \\
0.2 & 0.4 & 0.0 & 1:1 & 100 & 0.0006 & 0.0498 & 0.000 \\
0.2 & 0.4 & 0.0 & 1:1 & 200 & 0.0005 & 0.0503 & 0.000 \\
0.2 & 0.4 & 0.0 & 1:2 & 50 & 0.0015 & 0.0491 & 0.000 \\
0.2 & 0.4 & 0.0 & 1:2 & 100 & 0.0000 & 0.0496 & 0.000 \\
0.2 & 0.4 & 0.0 & 1:2 & 200 & -0.0001 & 0.0500 & 0.000 \\
0.2 & 0.4 & 0.2 & 1:1 & 50 & -0.0012 & 0.0496 & 0.000 \\
0.2 & 0.4 & 0.2 & 1:1 & 100 & -0.0013 & 0.0496 & 0.000 \\
0.2 & 0.4 & 0.2 & 1:1 & 200 & -0.0008 & 0.0496 & 0.000 \\
0.2 & 0.4 & 0.2 & 1:2 & 50 & 0.0002 & 0.0490 & 0.000 \\
0.2 & 0.4 & 0.2 & 1:2 & 100 & -0.0015 & 0.0501 & 0.000 \\
0.2 & 0.4 & 0.2 & 1:2 & 200 & -0.0010 & 0.0498 & 0.000 \\
0.2 & 1.0 & 0.0 & 1:1 & 50 & 0.0016 & 0.0493 & 0.000 \\
0.2 & 1.0 & 0.0 & 1:1 & 100 & 0.0004 & 0.0505 & 0.000 \\
0.2 & 1.0 & 0.0 & 1:1 & 200 & 0.0008 & 0.0508 & 0.000 \\
0.2 & 1.0 & 0.0 & 1:2 & 50 & -0.0010 & 0.0504 & 0.000 \\
0.2 & 1.0 & 0.0 & 1:2 & 100 & 0.0003 & 0.0490 & 0.000 \\
0.2 & 1.0 & 0.0 & 1:2 & 200 & 0.0007 & 0.0506 & 0.000 \\
0.2 & 1.0 & 0.2 & 1:1 & 50 & -0.0012 & 0.0506 & 0.000 \\
0.2 & 1.0 & 0.2 & 1:1 & 100 & 0.0005 & 0.0494 & 0.000 \\
0.2 & 1.0 & 0.2 & 1:1 & 200 & -0.0002 & 0.0488 & 0.000 \\
0.2 & 1.0 & 0.2 & 1:2 & 50 & -0.0006 & 0.0509 & 0.000 \\
0.2 & 1.0 & 0.2 & 1:2 & 100 & 0.0002 & 0.0494 & 0.000 \\
0.2 & 1.0 & 0.2 & 1:2 & 200 & -0.0005 & 0.0512 & 0.000 \\
0.4 & 0.4 & 0.0 & 1:1 & 50 & -0.0010 & 0.0501 & 0.000 \\
0.4 & 0.4 & 0.0 & 1:1 & 100 & -0.0005 & 0.0492 & 0.000 \\
0.4 & 0.4 & 0.0 & 1:1 & 200 & -0.0002 & 0.0484 & 0.000 \\
0.4 & 0.4 & 0.0 & 1:2 & 50 & -0.0009 & 0.0478 & 0.000 \\
0.4 & 0.4 & 0.0 & 1:2 & 100 & 0.0003 & 0.0498 & 0.000 \\
0.4 & 0.4 & 0.0 & 1:2 & 200 & 0.0003 & 0.0509 & 0.000 \\
0.4 & 0.4 & 0.2 & 1:1 & 50 & 0.0000 & 0.0493 & 0.008 \\
0.4 & 0.4 & 0.2 & 1:1 & 100 & -0.0013 & 0.0507 & 0.000 \\
0.4 & 0.4 & 0.2 & 1:1 & 200 & -0.0015 & 0.0504 & 0.000 \\
0.4 & 0.4 & 0.2 & 1:2 & 50 & -0.0005 & 0.0499 & 0.006 \\
0.4 & 0.4 & 0.2 & 1:2 & 100 & -0.0014 & 0.0500 & 0.000 \\
0.4 & 0.4 & 0.2 & 1:2 & 200 & -0.0012 & 0.0497 & 0.000 \\
0.4 & 1.0 & 0.0 & 1:1 & 50 & 0.0009 & 0.0491 & 0.000 \\
0.4 & 1.0 & 0.0 & 1:1 & 100 & -0.0002 & 0.0498 & 0.000 \\
0.4 & 1.0 & 0.0 & 1:1 & 200 & 0.0005 & 0.0510 & 0.000 \\
0.4 & 1.0 & 0.0 & 1:2 & 50 & 0.0007 & 0.0493 & 0.001 \\
0.4 & 1.0 & 0.0 & 1:2 & 100 & -0.0008 & 0.0491 & 0.000 \\
0.4 & 1.0 & 0.0 & 1:2 & 200 & 0.0000 & 0.0493 & 0.000 \\
0.4 & 1.0 & 0.2 & 1:1 & 50 & 0.0004 & 0.0502 & 0.000 \\
0.4 & 1.0 & 0.2 & 1:1 & 100 & 0.0000 & 0.0490 & 0.000 \\
0.4 & 1.0 & 0.2 & 1:1 & 200 & 0.0001 & 0.0490 & 0.000 \\
0.4 & 1.0 & 0.2 & 1:2 & 50 & -0.0015 & 0.0491 & 0.000 \\
0.4 & 1.0 & 0.2 & 1:2 & 100 & 0.0006 & 0.0493 & 0.000 \\
0.4 & 1.0 & 0.2 & 1:2 & 200 & -0.0004 & 0.0496 & 0.000 \\
\end{longtable}
\normalsize
\setlength{\tabcolsep}{6pt}
\renewcommand{\arraystretch}{1.0}

\scriptsize
\setlength{\tabcolsep}{3.0pt}
\renewcommand{\arraystretch}{1.05}
\begin{longtable}{cccccccc}
\caption{Detailed simulation results for the lognormal outcome distribution. Nonconvergence is reported as a percentage of the 100,000 simulation replicates.}\label{tab:supp-lognormal}\\
\toprule
\makecell{Efficacy-related\\missingness proportion} & \makecell{Dropout\\probability $p_d$} & \makecell{AE dropout\\probability} & \makecell{Allocation\\ratio} & $N$ & \makecell{Standardized\\bias} & \makecell{Test\\size} & \makecell{Nonconvergence\\(\%)} \\
\midrule
\endfirsthead
\multicolumn{8}{c}{Supplementary Table S2 (continued)}\\
\toprule
\makecell{Efficacy-related\\missingness proportion} & \makecell{Dropout\\probability $p_d$} & \makecell{AE dropout\\probability} & \makecell{Allocation\\ratio} & $N$ & \makecell{Standardized\\bias} & \makecell{Test\\size} & \makecell{Nonconvergence\\(\%)} \\
\midrule
\endhead
\midrule
\multicolumn{8}{r}{Continued on next page}\\
\endfoot
\bottomrule
\endlastfoot
0.0 & -- & 0.0 & 1:1 & 50 & 0.0000 & 0.0500 & 0.000 \\
0.0 & -- & 0.0 & 1:1 & 100 & -0.0003 & 0.0475 & 0.000 \\
0.0 & -- & 0.0 & 1:1 & 200 & -0.0004 & 0.0506 & 0.000 \\
0.0 & -- & 0.0 & 1:2 & 50 & 0.0005 & 0.0501 & 0.000 \\
0.0 & -- & 0.0 & 1:2 & 100 & -0.0008 & 0.0496 & 0.000 \\
0.0 & -- & 0.0 & 1:2 & 200 & -0.0007 & 0.0486 & 0.000 \\
0.0 & -- & 0.2 & 1:1 & 50 & -0.0003 & 0.0501 & 0.000 \\
0.0 & -- & 0.2 & 1:1 & 100 & -0.0002 & 0.0498 & 0.000 \\
0.0 & -- & 0.2 & 1:1 & 200 & 0.0005 & 0.0517 & 0.000 \\
0.0 & -- & 0.2 & 1:2 & 50 & 0.0005 & 0.0498 & 0.000 \\
0.0 & -- & 0.2 & 1:2 & 100 & -0.0004 & 0.0487 & 0.000 \\
0.0 & -- & 0.2 & 1:2 & 200 & -0.0009 & 0.0497 & 0.000 \\
0.2 & 0.4 & 0.0 & 1:1 & 50 & 0.0008 & 0.0485 & 0.000 \\
0.2 & 0.4 & 0.0 & 1:1 & 100 & 0.0002 & 0.0499 & 0.000 \\
0.2 & 0.4 & 0.0 & 1:1 & 200 & 0.0000 & 0.0495 & 0.000 \\
0.2 & 0.4 & 0.0 & 1:2 & 50 & -0.0010 & 0.0497 & 0.000 \\
0.2 & 0.4 & 0.0 & 1:2 & 100 & 0.0003 & 0.0503 & 0.000 \\
0.2 & 0.4 & 0.0 & 1:2 & 200 & -0.0002 & 0.0497 & 0.000 \\
0.2 & 0.4 & 0.2 & 1:1 & 50 & -0.0003 & 0.0505 & 0.000 \\
0.2 & 0.4 & 0.2 & 1:1 & 100 & -0.0015 & 0.0503 & 0.000 \\
0.2 & 0.4 & 0.2 & 1:1 & 200 & -0.0006 & 0.0494 & 0.000 \\
0.2 & 0.4 & 0.2 & 1:2 & 50 & -0.0020 & 0.0471 & 0.000 \\
0.2 & 0.4 & 0.2 & 1:2 & 100 & -0.0014 & 0.0499 & 0.000 \\
0.2 & 0.4 & 0.2 & 1:2 & 200 & -0.0003 & 0.0510 & 0.000 \\
0.2 & 1.0 & 0.0 & 1:1 & 50 & -0.0003 & 0.0484 & 0.000 \\
0.2 & 1.0 & 0.0 & 1:1 & 100 & -0.0002 & 0.0495 & 0.000 \\
0.2 & 1.0 & 0.0 & 1:1 & 200 & -0.0006 & 0.0507 & 0.000 \\
0.2 & 1.0 & 0.0 & 1:2 & 50 & 0.0007 & 0.0488 & 0.000 \\
0.2 & 1.0 & 0.0 & 1:2 & 100 & -0.0001 & 0.0497 & 0.000 \\
0.2 & 1.0 & 0.0 & 1:2 & 200 & 0.0004 & 0.0497 & 0.000 \\
0.2 & 1.0 & 0.2 & 1:1 & 50 & -0.0013 & 0.0486 & 0.000 \\
0.2 & 1.0 & 0.2 & 1:1 & 100 & -0.0004 & 0.0475 & 0.000 \\
0.2 & 1.0 & 0.2 & 1:1 & 200 & -0.0003 & 0.0503 & 0.000 \\
0.2 & 1.0 & 0.2 & 1:2 & 50 & 0.0006 & 0.0495 & 0.000 \\
0.2 & 1.0 & 0.2 & 1:2 & 100 & -0.0006 & 0.0483 & 0.000 \\
0.2 & 1.0 & 0.2 & 1:2 & 200 & -0.0008 & 0.0501 & 0.000 \\
0.4 & 0.4 & 0.0 & 1:1 & 50 & -0.0010 & 0.0481 & 0.001 \\
0.4 & 0.4 & 0.0 & 1:1 & 100 & 0.0004 & 0.0496 & 0.000 \\
0.4 & 0.4 & 0.0 & 1:1 & 200 & 0.0003 & 0.0502 & 0.000 \\
0.4 & 0.4 & 0.0 & 1:2 & 50 & 0.0009 & 0.0491 & 0.000 \\
0.4 & 0.4 & 0.0 & 1:2 & 100 & 0.0006 & 0.0493 & 0.000 \\
0.4 & 0.4 & 0.0 & 1:2 & 200 & 0.0010 & 0.0510 & 0.000 \\
0.4 & 0.4 & 0.2 & 1:1 & 50 & -0.0013 & 0.0478 & 0.002 \\
0.4 & 0.4 & 0.2 & 1:1 & 100 & -0.0002 & 0.0478 & 0.000 \\
0.4 & 0.4 & 0.2 & 1:1 & 200 & -0.0004 & 0.0489 & 0.000 \\
0.4 & 0.4 & 0.2 & 1:2 & 50 & -0.0015 & 0.0476 & 0.012 \\
0.4 & 0.4 & 0.2 & 1:2 & 100 & -0.0007 & 0.0487 & 0.000 \\
0.4 & 0.4 & 0.2 & 1:2 & 200 & -0.0007 & 0.0493 & 0.000 \\
0.4 & 1.0 & 0.0 & 1:1 & 50 & 0.0002 & 0.0485 & 0.000 \\
0.4 & 1.0 & 0.0 & 1:1 & 100 & 0.0007 & 0.0498 & 0.000 \\
0.4 & 1.0 & 0.0 & 1:1 & 200 & -0.0006 & 0.0495 & 0.000 \\
0.4 & 1.0 & 0.0 & 1:2 & 50 & -0.0002 & 0.0478 & 0.000 \\
0.4 & 1.0 & 0.0 & 1:2 & 100 & 0.0003 & 0.0488 & 0.000 \\
0.4 & 1.0 & 0.0 & 1:2 & 200 & -0.0002 & 0.0491 & 0.000 \\
0.4 & 1.0 & 0.2 & 1:1 & 50 & 0.0015 & 0.0486 & 0.000 \\
0.4 & 1.0 & 0.2 & 1:1 & 100 & 0.0002 & 0.0478 & 0.000 \\
0.4 & 1.0 & 0.2 & 1:1 & 200 & 0.0002 & 0.0501 & 0.000 \\
0.4 & 1.0 & 0.2 & 1:2 & 50 & 0.0015 & 0.0470 & 0.000 \\
0.4 & 1.0 & 0.2 & 1:2 & 100 & -0.0008 & 0.0483 & 0.000 \\
0.4 & 1.0 & 0.2 & 1:2 & 200 & 0.0002 & 0.0488 & 0.000 \\
\end{longtable}
\normalsize
\setlength{\tabcolsep}{6pt}
\renewcommand{\arraystretch}{1.0}

\end{document}